\documentclass{article}
\usepackage{amsthm,amssymb}
\usepackage[margin=1in]{geometry}
\usepackage{booktabs} %
\usepackage[ruled]{algorithm2e} %

\SetAlFnt{\small}
\SetAlCapFnt{\small}
\SetAlCapNameFnt{\small}
\SetAlCapHSkip{0pt}
\IncMargin{-\parindent}

\usepackage{natbib}

\usepackage{complexity}

\let\E\relax
\let\R\relax

\usepackage{bm}
\usepackage{fragkia}
\usepackage[colorlinks, citecolor=magenta, linkcolor=magenta]{hyperref}

\usepackage{tikz}
\usetikzlibrary{trees,positioning,shapes,calc}

\usepackage{xspace}
\usepackage{physics}
\usepackage[shortlabels]{enumitem}
\usepackage{bbm}
\usepackage[nameinlink]{cleveref}
\AtBeginDocument{
    \hypersetup{hypertexnames=false}
}
\newcommand{\tvx}{\tilde{\vx}}

\newcommand{\delimit}[3]{\newcommand{#1}[1]{\left#2##1\right#3}}
\delimit \ceil \lceil \rceil
\delimit \floor \lfloor \rfloor

\def\vb{{\bm{b}}}

\def\ve{{\bm{e}}}

\def\vh{{\bm{h}}}

\def\vs{{\bm{s}}}

\def\vu{{\bm{u}}}
\def\vw{{\bm{w}}}
\def\vx{{\bm{x}}}
\def\vy{{\bm{y}}}
\def\vz{{\bm{z}}}

\def\mA{{\mathbf{A}}}

\let\eps\epsilon
\let\op\operatorname
\newcommand{\ie}{\textit{i.e.}\xspace}
\newcommand{\eg}{\textit{e.g.}\xspace}
\newcommand\size{\op{size}}

\let\Root\varnothing

\let\ip\ev
\newcommand\sharpP{{\sf \#P}}

\usepackage{multirow}

\usepackage{pifont}%
\newcommand{\supp}{\op{supp}}
\newcommand\ind{\mathbbm 1}

\definecolor{briancolor}{rgb}{0, .5, 0}

\renewcommand{\vec}[1]{\bm{#1}}

\usepackage{nicefrac}

\newcommand{\CLS}{\mathsf{CLS}}
\newcommand{\FIXP}{\mathsf{FIXP}}

\title{The Complexity of Proper Equilibrium in Extensive-Form \\and Polytope Games\footnote{This paper contains and extends results that were originally in a prior version of \url{https://arxiv.org/abs/2511.03968}.}}

\usepackage{authblk}

\author[1]{Brian Hu Zhang}
\author[2]{Ioannis Anagnostides}
\author[2]{Kiriaki Fragkia}
\author[2]{\authorcr Maria-Florina Balcan}
\author[2,3]{Tuomas Sandholm}

\affil[1]{Massachusetts Institute of Technology}
\affil[2]{Carnegie Mellon University}
\affil[3]{Additional affiliations: Strategy Robot, Inc., Strategic Machine, Inc., Optimized Markets, Inc.}
\affil[ ]{}
\affil[ ]{\texttt{zhangbh}\texttt{@csail.mit.edu}, \texttt{\{ianagnos,kiriakif,ninamf,sandholm\}}\texttt{@cs.cmu.edu}}

\begin{document}

\begin{titlepage}

\maketitle
\pagenumbering{gobble}
\begin{abstract}
The \emph{proper equilibrium}, introduced by Myerson (1978), is a classic refinement of the Nash equilibrium that has been referred to as the ``mother of all refinements.'' For normal-form games, computing a proper equilibrium is known to be $\mathsf{PPAD}$-complete for two-player games and $\mathsf{FIXP}_a$-complete for games with at least three players. However, the complexity beyond normal-form games---in particular, for \textit{extensive-form games (EFGs)}---was a long-standing open problem first highlighted by Miltersen and Sørensen (SODA '08).

In this paper, we resolve this problem by establishing $\mathsf{PPAD}$- and $\mathsf{FIXP}_a$-membership (and hence completeness) of normal-form proper equilibria in two-player and multi-player EFGs respectively. Our main ingredient is a technique for computing a perturbed (proper) best response that can be computed efficiently in EFGs. This is despite the fact that, as we show, computing a best response using the classic perturbation of Kohlberg and Mertens based on the permutahedron is $\mathsf{\#P}$-hard even in Bayesian games. 

In stark contrast, we show that computing a proper equilibrium in \textit{polytope} games is $\mathsf{NP}$-hard. This marks the first natural class in which the complexity of computing equilibrium refinements does not collapse to that of Nash equilibria, and the first problem in which equilibrium computation in polytope games is strictly harder---unless there is a collapse in the complexity hierarchy---relative to extensive-form games.
\end{abstract}

\newpage
\tableofcontents

\end{titlepage}
\pagenumbering{arabic}

\section{Introduction}

The Nash equilibrium is the cornerstone of noncooperative game theory~\citep{Nash50:Equilibrium}. Its formulation has had a profound impact in economics, which, according to~\citet{Myerson99:Nash}, is comparable to the discovery of the DNA double helix in the biological sciences. Despite this stature, an early critique of the Nash equilibrium is that it can prescribe highly suboptimal strategies away from the equilibrium path, which can arise if another player deviates from equilibrium play by way of a mistake---a ``tremble.'' This casts doubt on both its descriptive validity and its prescriptive power. To address these shortcomings, a vibrant theory of \emph{equilibrium refinements} was developed starting in the 1970s, spearheaded by Nobel prize-winning economists such as~\citet{Selten75:Reexamination}~and~\citet{Myerson78:Refinements}. The central aim of this line of work is to refine the solution space by selecting more robust, higher quality equilibria. Indeed, equilibrium refinements have been shown to produce better strategies in practice~\citep{Farina18:Practical,Kubicek26:Equilibrium,Bernasconi24:Learning}.

Among the elaborate hierarchy of equilibrium refinements, Myerson's \emph{proper equilibrium} stands out as a particularly influential and appealing concept. Introduced almost 50 years ago, it is based on a compelling behavioral axiom: while players may tremble and play suboptimal strategies, they are less likely to make costlier mistakes. Peter Bro Miltersen has anecdotally referred to proper equilibria as ``the mother of all refinements~\citep{Etessami14:Complexity}.''

Despite its prominence for many decades, the computational properties of proper equilibria have remained elusive. In the special case of normal-form games, the complexity landscape is by now well-understood: finding a proper equilibrium is \PPAD-complete for two-player or polymatrix games and $\FIXP_a$-complete for games with at least three players~\citep{Sorensen12:Computing,Hansen18:Computational}. However, from a computational standpoint, the normal-form representation is ill suited in modeling most strategic interactions, which feature sequential moves and imperfect information. In such settings, casting the game in normal form typically leads to an exponential blow up, so the previous complexity results do not apply. The canonical, compact representation in sequential settings is the \emph{extensive form}. The complexity of computing proper equilibria in extensive-form games has remained an outstanding open problem for two decades, first explicitly mentioned by~\citet{Miltersen08:Fast} and more recently highlighted by~\citet{Sorensen12:Computing} for the case of two-player games.

There is a daunting challenge in the computation of proper equilibria beyond normal-form games: the classic approach, introduced by~\citet{Kohlberg86:strategic}, involves a perturbation scheme that essentially requires sorting the pure strategies by their expected utility. This is hardly a challenge in normal-form games, but in extensive-form games---and more broadly \emph{polytope games}---each player has a number of strategies exponential in the representation of the game.\footnote{Such games are often referred to as games of \emph{exponential type}~\citep{Papadimitriou08:Computing}.} Virtually all algorithmic results concerning proper equilibria---with the one exception of~\citet{Miltersen06:Computing} in the special case of (two-player) zero-sum games---rely on the permutahedron per~\citet{Kohlberg86:strategic}.

\subsection{Our results} 

Our main contribution in this paper is to settle the complexity of computing normal-form proper equilibria in extensive-form games, establishing \PPAD-completeness in two-player games and $\FIXP_a$-completeness in games with at least three players. 

\paragraph{The permutahedron barrier} As discussed, a key algorithmic challenge is that using the standard perturbation of~\citet{Kohlberg86:strategic} based on the permutahedron, one would naively have to sort exponentially many vertices, which is clearly hopeless. We formalize this barrier, showing $\sharpP$-hardness for computing a perturbed best response \emph{\`a la} Kohlberg and Mertens.

\begin{theorem}[Informal; formal version in~\Cref{theorem:hardness-KM}]
    Computing a Kohlberg-Mertens best-response is $\sharpP$-hard even in Bayesian games.
\end{theorem}

\paragraph{Complexity in polytope games} Second, we show that in general {\em polytope games}, there is no hope of $\PPAD$ or $\FIXP$ membership, barring a significant breakthrough in complexity theory: it is $\NP$-hard to compute a proper equilibrium of a polytope game, even with two players.\footnote{A problem in $\PPAD$ cannot be \NP-hard unless $\NP = \coNP$~\citep{Johnson88:How}. Also, a problem in $\FIXP$ cannot be \NP-hard unless $\coNP \subseteq \NP_{\mathbb{R}}$, where $\NP_{\mathbb{R}}$ is the class of problems polynomial-time reducible to deciding the existential theory of the reals~\citep{Etessami07:Complexity}.}

\begin{theorem}[Informal; formal version in~\Cref{th:nfpe-hard}]
Computing a proper equilibrium of a polytope game, described explicitly though linear constraints, is $\NP$-hard.
\end{theorem}

Nonetheless, in our main result (\Cref{sec:efg br,sec:efg-proper}), we show that $\PPAD$- and $\FIXP$-membership {\em does} hold for {\em extensive-form games}. Given the negative result for polytope games, this should come as a surprise: to our knowledge, for all other equilibrium concepts whose computation has been studied so far, the complexity for extensive-form games matches (up to polynomial factors) the complexity for polytope games. 

\paragraph{Special case: the hypercube}

We start by describing our construction in the special case of the hypercube. The analysis in this subsection will be subsumed by the more general case of extensive-form games in \Cref{sec:efg br}; however, we include it because the hypercube is both complex enough to illustrate the main idea of the construction, and simple enough to be significantly easier to understand than the general extensive-form case.

Let $\cS = \{0, 1\}^d$ be the unit $d$-dimensional discrete hypercube, and $\cX = \conv \cS = [0, 1]^d$. Fix a utility vector $\vu \in \R^d$. To compute proper equilibria, a useful subroutine is to compute $\eps$-proper best responses. More formally, a point $\vx \in \cX$ is an $\eps$-proper best response if there exists some {\em full-support} distribution $\lambda \in \Delta(\cS)$ such that
\begin{enumerate}
    \item $\E_{\vs \sim \lambda} \vs = \vx$, and
    \item if $\ip{\vu, \vs - \vs'} < 0$, \ie, if $\vs$ has worse utility than $\vs'$, then $\lambda(\vs) \le \eps \cdot \lambda(\vs')$.
\end{enumerate}
A naive algorithm for computing an $\eps$-proper best response would be to sort all $2^d$ vertices of $\cS$, then construct a distribution $\lambda$ satisfying the above property (which is straightforward once the vertices have been sorted), and output its expectation. However, this naive algorithm, of course, is not efficient. We now illustrate a {\em polynomial-time} algorithm for computing $\eps$-proper best responses.

For notation in this section, we will use $\ve_i \in \cS$ to denote the $i$th standard basis vector, and more generally $\ve_A$ for a set $A \subseteq [d]$ to denote the vector with $1$s at the indices in set $A$. By symmetry, we may assume without loss of generality that the entries of $\vu$ are sorted and all nonpositive:
\begin{align*}
    \vu[d] \le \vu[d-1] \le \dots \vu[1] \le 0.
\end{align*}
Then we claim that 
\begin{align}
    \vx := (\hat\eps, \hat\eps^2, \dots, \hat\eps^d) \in [0, 1]^d, \label{eq:hypercube-br}
\end{align}
where $\hat\eps$ is a value to be chosen later, is an $\eps$-proper best response. To see this, consider the following distribution $\lambda$: $\lambda$ assigns probability $\vx[i]$ to each $\ve_i \in \R^d$, and all remaining probability on $\vec 0$. 

It is easy to check that the expectation of $\lambda$ is $\vx$. However, $\lambda$ is not full-support, so we need to modify $\lambda$ so that it is. The critical observation is the following: we can exploit the affine dependencies in $\cS$ to modify the probabilities in $\lambda$ without changing the expectation of $\lambda$. 

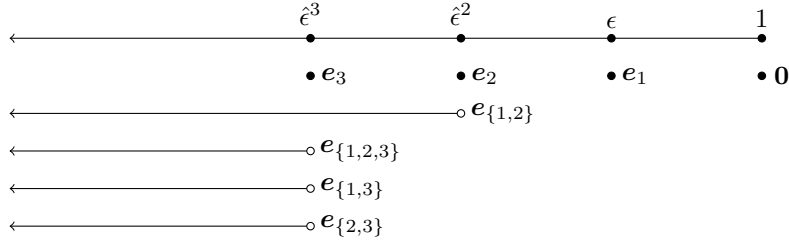
\begin{figure}
    \centering
    \tikzset{
        every node/.style={
            circle, draw, fill=black, inner sep=1pt
        },
        open/.style={fill=white},
    }
    \begin{tikzpicture}[x=20mm, y=5mm]
        \draw[->] (0, 0) -- (-5, 0);
        \node[label=above:{1}, ] at (0, 0){};
        \node[label=above:{$\eps$}, ] at (-1, 0){};
        \node[label=above:{$\hat\eps^2$}, ] at (-2, 0){};
        \node[label=above:{$\hat\eps^3$}, ] at (-3, 0){};
        \node[label=right:{$\vec 0$}, ] at (0, -1){};
        \node[label=right:{$\ve_1$}, ] at (-1, -1){};
        \node[label=right:{$\ve_2$}, ] at (-2, -1){};
        \node[label=right:{$\ve_3$}, ] at (-3, -1){};
        \draw[->] (-2, -2) -- (-5, -2);
        \node[label=right:{$\ve_{\{1, 2\}}$}, open] at (-2, -2){};
        \draw[->] (-3, -4) -- (-5, -4);
        \node[label=right:{$\ve_{\{1, 3\}}$}, open] at (-3, -4){};
        \draw[->] (-3, -5) -- (-5, -5);
        \node[label=right:{$\ve_{\{2, 3\}}$}, open] at (-3, -5){};
        \draw[->] (-3, -3) -- (-5, -3);
        \node[label=right:{$\ve_{\{1, 2, 3\}}$}, open] at (-3, -3){};
    \end{tikzpicture}
    \caption{Illustration of the proper best response construction over the hypercube, for $d=3$. The number line is on a logarithmic scale. The probability masses of $\vec 0, \ve_1, \ve_2$, and $\ve_3$ are fixed (up to lower-order terms) to the powers of $\hat \eps$, while the masses of $\ve_A$ for $|A| > 1$ can be changed to be anywhere in the indicated range of values. }
    \label{fig:hypercube}
\end{figure}

For example, consider the point $\ve_{\{1, 2\}}$. It satisfies the linear constraint
$
    \ve_{\{1, 2\}} = \ve_1 + \ve_2 - \vec 0.
$
Therefore, if we add probability mass $\delta > 0$ to $\ve_{\{1, 2\}}$, we can maintain the expectation of $\lambda$ by adding the same mass $\delta$ to $\vec 0$, and removing it from $\ve_1$ and $\ve_2$. How large can we make $\delta$? Since $\ve_1$ and $\ve_2$ currently both have mass $\hat\eps^2$, we can certainly make $\delta$ any number asymptotically smaller than $\hat\eps^2$ (for example, $\hat \eps^{2 + 2^{-d}}$) while maintaining the condition that $\lambda$ is a valid distribution (for sufficiently small $\hat\eps$), and only changing the masses of $\ve_1$ and $\ve_2$ by an asymptotically negligible amount. More generally, for any vertex $\ve_A$, we can use the linear relationship
\begin{align*}
    \ve_A = \vec 0 + \sum_{i \in A} (\ve_i - \vec 0)
\end{align*}
to add mass asymptotically smaller than $\hat \eps^{\max(A)}$ to $\ve_A$. Notice that $\ve_A$ has weakly worse utility than $\ve_i$ for every $i \in A$. Since this is the only constraint on the mass of $\ve_A$, by adding appropriate amounts of mass, all the vertices $\ve_A$ can be placed in ``sorted order'', in the required sense: if $\ip{\vu, \ve_A - \ve_{A'}} < 0$ then $\lambda(\ve_A) \ll \lambda(\ve_{A'})$, where the $\ll$ indicates that there is a multiplicative gap (say, approximately $\hat\eps^{2^{-d}}$) between any two distinct values of $\lambda$. A visualization of the possible values of the masses for $d=3$ can be found in \Cref{fig:hypercube}. Thus, to satisfy the proper best response condition, it suffices to use this construction and then take $\hat\eps \sim \eps^{2^d}$.

Note that, crucially, although the above {\em proof technique} relies on sorting all the vertices $\vs \in \cS$, the {\em algorithm for computing the proper best response} does not: the algorithm only needs to sort the indices $i \in [d]$, then output the vector as defined in \eqref{eq:hypercube-br}, which can be done efficiently!

Turning this into an algorithm for proper equilibrium computation in general extensive-form games requires several more steps:
\begin{enumerate}
    \item Extensive-form games, while containing the hypercube as a special case, are far more expressive. So, we must generalize the above construction. Intuitively, the formalization relies on generalizing the observation that it suffices for the proper equilibrium condition to apply at the level of {\em coordinates}: 
    \begin{align}
       \vu[i] < \vu[j] \qq{$\implies$} \vx[i] \le \hat\eps \cdot \vx[j]. \label{eq:hypercube-seq-proper}
    \end{align}
    \item For positive results, we appeal to the optimization framework of \citet{FilosRatsikas21,FilosRatsikas24:PPAD}. To use this framework, the best-response oracle (BRO) must have a particular form, namely, it must be a linear feasibility problem with a particular structure. The system of constraints \eqref{eq:hypercube-seq-proper}, generalized to the extensive-form setting, will yield a BRO of the correct form. %
    
    \item If there are multiple best responses, \citet{FilosRatsikas21,FilosRatsikas24:PPAD} requires that the BRO be allowed to return {\em any} solution to the linear feasibility problem. In our context, this means that $\vec x$ may be {\em any} vector satisfying \eqref{eq:hypercube-seq-proper}; in particular, it is possible for $\sum_i \vx[i] \ge 1$, so that the initial construction of $\lambda$, in which we set the probability of $\ve_i$ to $\vx[i]$, is not valid. We thus need to be more careful in constructing the inital distribution $\lambda$. We choose instead to take $\lambda$ to be the product distribution that assigns probability $\prod_{i \in A} \vx[i] \cdot \prod_{i \notin A} (1 - \vx[i])$ to each vertex $\ve_A$. The resulting analysis is more intricate than the simplified version shown in this section. 
    \item The above construction requires $\hat\eps \approx \eps^{2^d}$, which cannot be efficiently represented! Fortunately, we are interested in computing {\em proper} equilibria, which are defined as the {\em limit points} of $\eps$-proper equilibria---and, as it turns out, computing a limit point will not require $\hat\eps$ to be so small.
\end{enumerate}
We address all of these issues in \Cref{sec:efg br,sec:efg-proper}, leading to our main result:
\begin{theorem}[Main result, informal; formal version in \Cref{th:main}]
    Computing a proper equilibrium of an extensive-form game has the same complexity as computing a Nash equilibrium.
\end{theorem}

\subsection{Related work}

A key reference point to our work is the paper of~\citet{Hansen18:Computational}, which centers on the complexity of computing proper equilibria in normal-form games. They showed that a \emph{strong approximation}---in the sense of~\citet{Etessami07:Complexity}---of proper equilibrium in multi-player games is $\FIXP_a$-complete. For polymatrix games, they established \PPAD-completeness, extending an earlier result by~\citet{Sorensen12:Computing} pertaining to two-player general-sum games. An important nuance about the $\PPAD$ membership result of~\citet{Hansen18:Computational} is that it concerns the computation of \emph{symbolic strategies}. This is important in view of the fact that even \emph{verifying} whether a strategy profile is a proper equilibrium in a two-player game is \NP-hard; one can circumvent this hardness by working with symbolic strategies. In this paper, following the original formulation of~\citet{Johnson88:How} on total search problems, we work with a definition of $\PPAD$ that does not require efficient verification (\Cref{remark:verif}).

The complexity of perfect equilibria is by now well-understood even in extensive-form games. Early pioneering work by~\citet{Hansen10:Computational} and~\citet{Etessami14:Complexity} characterized the complexity in normal-form games. The landscape in extensive-form games is more intricate. While one can revert to the usual notion of perfection in the induced normal-form game---known as \emph{normal-form perfect equilibrium}---there are significantly stronger solution concepts~\citep{Etessami21:Complexity}. Two notable such notions are \emph{extensive-form perfect equilibria (EFPEs)} and \emph{quasi-perfect equilibria (QPEs)}~\citep{VanDamme91:Stability}, which are incomparable~\citep{Mertens95:Two}. The complexity of these refinements has been the subject of much research~\citep{Gatti20:Characterization,Miltersen10:Computing,Farina17:Extensive}, culminating in the paper of~\citet{Etessami21:Complexity}.

Normal-form proper equilibria refine quasi-perfect equilibria, as was famously shown by~\citet{vanDamme84:relation}. We elaborate on the difference between normal-form proper and quasi-perfect equilibria in~\Cref{sec:qpes}. In fact, \citet{vanDamme84:relation} showed that normal-form proper equilibria refine a subset of quasi-perfect equilibria known as \emph{quasi-proper equilibria}. (\Cref{sec:diagram} provides a schematic illustration to help the reader keep track of the different refinements.) \citet{Hansen21:Computational} characterized the complexity of quasi-proper equilibria in extensive-form games, establishing \PPAD-completeness in two-player games and $\FIXP_a$-completeness in multi-player general-sum games. Prior to our work, the only result about normal-form proper equilibria we are aware of in extensive-form games concerns (two-player) zero-sum games~\citep{Miltersen08:Fast}, where a polynomial-time algorithm was developed. Interestingly, their algorithmic approach does not rely on the permutation scheme of~\citet{Kohlberg86:strategic} using the permutahedron. However, their approach is tailored to the zero-sum setting. It is worth noting that our algorithmic approach in this paper provides an alternative polynomial-time algorithm for computing normal-form proper equilibria in zero-sum extensive-form games.

\section{Preliminaries}

In this section, we introduce relevant background on equilibrium refinements and extensive-form games.

\begin{definition}
A (finite) {\em normal-form game} consists of the following elements.
\begin{enumerate}
    \item a finite set of players $[n] := \{ 1, \dots, n\}$;
    \item for each player $i$, a finite set of {\em pure strategies} $\cS_i$; 
    \item for each player $i$, a \textit{utility function}  $u_i: \cS \to \R$, where $\cS := \bigtimes_{i=1}^n \cS_i$.
\end{enumerate}
\end{definition}
A tuple $\vs := (s_1, \dots, s_n) \in \cS$ is called a {\em pure profile}. As is standard in game theory, we will use $\vs_{-i} \in \cS_{-i} := \bigtimes_{j \ne i} \cS_j$ to denote a profile for all players except $i$.
Players are also allowed to select {\em mixed strategies} $\lambda_i \in \Delta(\cS_i)$. A mixed strategy $\lambda_i$ is called {\em fully mixed} if it assigns strictly positive probability to every pure strategy. A tuple of mixed strategies $\vlambda = (\lambda_1, \dots, \lambda_n) \in \Delta(\cS_1) \times \dots \times \Delta(\cS_n)$ is a {\em mixed profile}. Finally, we will overload the notation
$
    u_i(\vlambda) = \E_{\vs \sim \vlambda} [u_i(\vs)],
$
to denote the expected utility of a mixed strategy profile.

In games, we are typically concerned with notions of {\em equilibrium}, which informally are selections of strategies for each player such that no player wishes to change her strategy. Here, we define the notions of equilibrium that will be relevant to our paper. 

\begin{definition}
    A strategy $\lambda_i$ is a {\em best response} to a profile $\vlambda_{-i}$ if $u_i(s_i, \vlambda_{-i}) \le u_i(\vlambda)$ for every possible strategy $s_i \in \cS_i$. A strategy profile $\vlambda$ is a {\em Nash equilibrium} if every player's strategy is a best response.
\end{definition}

To define equilibrium refinements such as perfect and proper equilibria, we need to define the notion of $\eps$-perfect and $\eps$-proper equilibria.

\begin{definition}
    A strategy $\lambda_i$ is an $\eps$-{\em perfect best response} to a profile $\vlambda_{-i}$ if it is fully mixed, and for all pairs of pure strategies $s_i, s_i' \in \cS_i$, we have $$u_i(s_i, \vlambda_{-i}) < u_i(s_i', \vlambda) \qq{$\Rightarrow$} \lambda_i(s_i) \le \eps.$$ That is, every suboptimal strategy is played with probability at most $\eps$.
\end{definition}

\begin{definition}
    A strategy $\lambda_i$ is an $\eps$-{\em proper best response} to a profile $\vlambda_{-i}$ if it is fully mixed, and for all pairs of pure strategies $s_i, s_i' \in \cS_i$, we have $$u_i(s_i, \vlambda_{-i}) < u_i(s_i', \vlambda_{-i}) \qq{$\Rightarrow$} \lambda_i(s_i') \le \eps \cdot \lambda_i(s_i).$$
\end{definition}
Obviously, every $\eps$-proper best response is $\eps$-perfect, but not vice-versa. We are now ready to define perfect and proper equilibria.
\begin{definition}
    An {\em $\eps$-perfect equilibrium} is a fully-mixed profile $\vlambda$ such that every player is playing an $\eps$-perfect best response. A {\em perfect equilibrium} is a limit point of $\eps$-perfect equilibria as $\eps \to 0$. Formally, $\vlambda^*$ is a perfect equilibrium if there are infinite sequences $\{ \vlambda^{(t)}\}_{t=1}^\infty \to \vlambda^*$ and $\{ \eps^{(t)} \}_{t=1}^\infty \to 0$ such that $\vlambda^{(t)}$ is a $\eps^{(t)}$-perfect equilibrium for every $t$.

    We define {\em ($\eps$-)proper equilibria} identically, by replacing every instance of {\em perfect} with {\em proper}.
\end{definition}

\subsection{Polytope and extensive-form games}

A {\em polytope game} is a special kind of normal-form game that can be expressed concisely. More formally, for each player, let  $\cX_i \subset \R^{d_i}$ be a convex compact polytope, and $\cX = \bigtimes_{i=1}^n \cX_i$ for the set of strategy profiles. We will assume that each player's utility $u_i : \cX \to \R$ is {\em multilinear}, that is, we will assume that $u_i(\vx_j, \vx_{-j})$ is an affine function of $\vx_j$ for every player $j$ and profile $\vx_{-j} \in \cX_{-j}$. We will treat polytope games as concise representations of normal-form games, in the sense that a polytope game represents the normal-form game in which each player's set of pure strategies $\cS_i$ is the set of vertices of $\cX_i$. %

The definitions of equilibria for normal-form games are defined over the sets of {\em mixed} strategies $\Delta(\cS_i)$. Since $\cS_i$ can be exponentially large, mixed strategies cannot be represented concisely. However, for polytope games, by multilinearity of the utility, it is easy to see that the utilities of each player depend only on the expectations: for example, whether $\vlambda_i$ is a Nash equilibrium depends only on the point $\E_{\vs \sim \vlambda}\vs \in \cX$. Thus, we will overload the definitions of equilibria and best responses to also apply to points $\vx \in \cX$ in the polytope; for example, we will say that $\vx$ is a Nash, perfect, or proper equilibrium if $\vx = \E_{\vs \sim \vlambda}\vs$ for some Nash, perfect, or proper equilibrium $\vlambda$.\footnote{Technically, moving from distributions to their expectations incurs a loss of information: for any given $\vx$ there may be many distributions $\lambda$ whose expectation is $\vx$. However, this choice of representation will allow us to circumvent the trivial  hardness barrier of needing to represent exponentially many probabilities to specify a probability distribution on $\cS$.}

{\em Extensive-form games} (EFGs) are a special class of polytope games that model games with imperfect information and multiple timesteps. In an extensive-form game, each player's strategy set $\cX_i$ arises from a {\em tree-form decision problem}. The tree-form decision problem for player $i$ is represented as a rooted tree with two types of nodes: {\em decision points} $j \in \cJ_i$, and \emph{observation points} (or \emph{sequences}), denoted as $\sigma \in \Sigma_i$. We assume WLOG that the root, denoted as $\varnothing \in \Sigma_i$, is always an observation point; that observation and decision points alternate along every root-to-leaf path, and that all leaves are also observation points. Outgoing edges from decision points are labeled with distinct {\em actions} $a$. The set of actions at decision point $j$ is denoted $\cA_j$, and for a given action $a \in \cA_j$ we will use $ja$ to denote the observation point child reached by following action $a$. The set of decision points that are children of a given observation point $\sigma$ is $C_\sigma$.  The number of observation points is $d_i := |\Sigma_i|$. The parent observation point of a decision point $j$ is denoted $p_j$. Finally, the tree-form decision problem induces a partial order $\succeq$, \eg, $\Root \succeq s$ for every non-root node $s$.

A \emph{behavioral strategy} $\beta_i$ for player $i$ is a collection of distributions one for each decision point. That is,  $\beta_i(\cdot|j) \in \Delta(\cA_j)$ denotes the probability distribution from which player $i$ samples their strategy. 
The {\em sequence form} of $\beta_i$ is the vector $\vx_i \in \R^{d_i}$, where $\vx_i[\sigma] = \prod_{ja \preceq \sigma} \beta_i(a|j)$ denotes the product of probabilities on the path to sequence $\sigma \in \Sigma_i$. The strategy set $\cX_i$ of player $i$ is then the set of all possible sequence form strategies for player $i$.  \citet{Romanovskii62:Reduction} and \citet{Stengel96:Efficient} showed that $\cX_i$ for player $i$ can be described by a system of linear constraints
\begin{align*}
    \vx_i[\varnothing] = 1, \;\;\; \vx_i[p_j] = \sum_{a \in A_j} \vx_i[ja] \;\; \forall j \in \cJ_i, \;\;\; \vx_i[\sigma] \geq 0 \;\; \forall \sigma \in \Sigma_i.
\end{align*}

\begin{remark}[Equilibrium refinement concepts for EFGs]
    In this paper, our focus is on computing {\em normal-form} proper equilibria of polytope games, including extensive-form games. A distinct notion of perfect refinement can be defined based on the extensive-form game tree; this is known as the {\em extensive-form perfect equilibrium} (EFPE). Surprisingly, the sets of extensive-form perfect and normal-form perfect equilibria may even be {\em disjoint}~\cite{Mertens95:Two}, so the results of the present paper do not have any implications for EFPE; indeed, results analogous to those in our paper for EFPE are known~\cite{Etessami21:Complexity}.
\end{remark}

\subsection{Representation and computational considerations}

Real-valued functions will be represented via arithmetic circuits, where the valid gates are $+,-, \times, \max, \min$, and rational constants. Such a circuit, by construction, always encodes a continuous function.  An arithmetic circuit is {\em piecewise linear} if all multiplication gates are multiplications by a rational constant. An arithmetic circuit defining a utility function $u_i : \cX_1 \times \dots \times \cX_n \to \R$ of a game is {\em multilinear} if it only uses gates $+, \times$, and rational constants, and different inputs to any given multiplication gate must involve disjoint sets of inputs and rational constants.

A generic arithmetic circuit may not be efficiently computable: for example, one can write down an arithmetic circuit that computes the number $2^{2^n}$ by repeated squaring in $n$ gates. However, both the multilinearity and piecewise linearity conditions prohibit repeated squaring, and thus allow efficient computation of the resulting circuits~\cite{Anonymous26:Complexity,Fearnley23:Complexity}.

\subsubsection{Complexity classes}

The complexity classes we will be interested in are classes involving {\em search problems}. A {\em search problem} is a function $R : \Sigma^* \to 2^{\Sigma^*}$, where $\Sigma^*$ is the set of all finite strings over a fixed alphabet $\Sigma$, \eg, $\Sigma = \{0, 1\}$. A search problem is {\em total} if $R(x) \ne \emptyset$ for all $x$. An algorithm {\em solves} a search problem if, given $x \in \Sigma^*$, it outputs some $y \in R(x)$. If $R$ and $R'$ be two total search problems, we say that $R$ {\em polynomial-time reduces to} $R'$ if there are polynomial-time computable functions $P$ and $P'$ such that $P'(R'(P(x))) \subseteq R(x)$ for all $x \in \Sigma^*$. Given these definitions, we may define the complexity classes by defining complete problems for each of them, following \citet{Etessami07:Complexity}.

\begin{definition}[$\FIXP_a$]
    The class $\FIXP_a$ consists of all total search problems that polynomial-time reduce to the following problem: given a function $f : [0,1]^n \to [0,1]^n$ as an arithmetic circuit, and an accuracy parameter $\delta > 0$ (as a rational number in binary), find a point $\vx \in [0,1]^n$ such that $\norm{\vx - \vx^*}_2 \le \delta$, where $\vx^*$ is a fixed point of $f$.
\end{definition}

\begin{definition}[$\PPAD$]
    The class $\PPAD$ consists of all total search problems that polynomial-time reduce to the following problem: given a function $f : [0,1]^n \to [0,1]^n$ as a {\em piecewise linear} arithmetic circuit, find an {\em exact} fixed point $\vx^*$ of $f$.
\end{definition}

Brouwer's fixed point theorem guarantees that both problems are total, and piecewise linearity guarantees the existence of an efficiently-representable exact fixed point~\cite{Etessami07:Complexity} in the latter case.

\begin{remark}[Solution verification]
    \label{remark:verif}
    It is common to include in the definition of $\PPAD$ the stipulation that, for a search problem $R$ to be in $\PPAD$, all solutions must be {\em efficiently verifiable}, that is, given $x$ and $y$, there should be a polynomial-time algorithm for deciding whether $y \in R(x)$. Following \citet{Johnson88:How}, we {\em intentionally} do not stipulate this condition. Indeed, checking if a point $\vx \in \cX$ is a perfect or proper equilibrium is $\NP$-hard even for normal-form games~\cite{Hansen10:Computational}. For further discussion on this point, we refer to \citet{Anonymous26:Complexity}.
\end{remark}

\subsubsection{Polytope games}
We will assume that polytope games are given as follows. Each player's strategy set $\cX_i$ is encoded via an explicit set of bounding linear constraints. The utility functions $u_i : \cX \to \R$ are encoded via multilinear arithmetic circuits. 

We will be interested in the problem of computing $\delta$-{\em near} equilibria. That is, given a game $\Gamma$, a precision $\delta > 0$, and a solution concept (Nash, perfect, or proper equilibria), our goal is to output a point $\vx \in \cX$ such that $\norm{\vx - \vx^*}_2 \le \delta$ for some exact equilibrium $\vx^*$.

\begin{remark}[Almost vs. near equilibria]
A different definition of approximate Nash equilibrium%
would arise from insisting that every player is $\delta$-almost best responding, that is, $u_i(\vx_i', \vx_{-i}) \le u_i(\vx) + \delta$ for every player $i$ and deviation $\vx_i' \in \cX_i$. We call these $\delta$-{\em almost} equilibria, to distinguish this from $\delta$-near equilibria. However, the definition of $\delta$-almost equilibria cannot be reasonably translated to perfect and proper equilibria, because these refinements are defined as {\em limit points} of certain sequences of approximate equilibria. We are therefore essentially forced to discuss {\em near} equilibria, and we will not discuss {\em almost} further equilibria in this paper. The choice to use near-equilibria rather than almost equilibria when discussing equilibrium refinements is standard: for example, \citet{Etessami14:Complexity,Etessami21:Complexity,Hansen10:Computational,Anonymous26:Complexity} make the same choice for the same reasons.
\end{remark}

\section{Hardness of proper equilibria in polytope games}\label{sec:hardness}

For most of the remainder of the paper, we will drop the subscript $i$ indicating the player, since we will usually be interested in computing a single-player best response, as this is the most difficult and intricate part of our results. Let $\cX$ be a convex compact polytope and $\cS$ be its set of vertices. Let $\vu$ be a utility vector. Our goal is to compute a {\em perturbed best response} to $\vu$ that is valid for proper equilibria. We will use $\vx \succeq \vx'$  to mean $\ip{\vu, \vx} \ge \ip{\vu, \vx'}$; the utility vector $\vu$ under consideration will always be clear from context. Formally, our goal is to compute the following.
\begin{definition}[$\eps$-proper best response]\label{def:eps-proper-br}
Given a convex compact polytope $\cX \subset \R^d$ and $\vu \in \R^d$, an $\eps$-proper best response is a strategy $\vx^* \in \cX$ such that $\vx^* = \E_{\vs\sim\lambda}\vs$
    for {\em some} full-support distribution $\lambda \in \Delta(\cS)$ such that $\vs \prec \vs'$ implies $\lambda(\vs) \le \eps \cdot \lambda(\vs')$.
\end{definition}

There is an immediate barrier to computing $\eps$-proper best responses and hence $\eps$-proper equilibria: the entries of $\vx^*$ may need to be doubly exponentially small. For example, take $\eps = 1/2$, $\cS = \{0,1\}^d$ and $\vu = -(1, 2, 4, \dots, 2^{d-1})$. All strategies $\vs \in \cS$ with $\vs[d] = 1$ are worse than all strategies with $\vs[d] = 0$; since there are no ties, this implies that $0 < \vx^*[d]  \le 2^{d-1}/2^{2^{d-1}}$ in any $\eps$-proper best response. But any positive rational number less than $2^{d-1}/2^{2^{d-1}}$ requires exponentially many bits to represent, so no $\eps$-proper best response is efficiently representable.

There are at least two possible resolutions to this problem. The first is to only consider the problem of computing a {\em symbolic} $\eps$-proper best response, where our output $\vx^*$ is a polynomial in $\eps$; in this polynomial expression, the exponents on $\eps$ are represented in binary, so expressions such as $\eps^{2^d}$ are legal. The second is to skip computational considerations of $\eps$-proper best responses altogether, and directly consider the problem of computing a proper equilibrium, that is, the {\em limit point}. We will briefly discuss the former approach, and then focus mostly on the latter.

The standard approach to computing an $\eps$-proper best response would be to sort the pure strategies $\vs_1, \dots, \vs_N \in \cS$ in decreasing order of utility, $\vs_1 \succeq \vs_2 \succeq \dots \succeq \vs_N$, define the distribution $\lambda$ that assigns probability proportional to $\eps^i$ to strategy $\vs_i$, then compute and output the expectation $\vx$ of $\lambda$. We refer to this best response as the {\em Kohlberg-Mertens (KM)} best response, because it is the one suggested by \citet{Kohlberg86:strategic}. This naive approach is already sufficient for normal-form games (when $N$ is small), but fails for polytope games, because $N$ can be exponentially large compared to the size of the input. Indeed, our first result shows that even computing a symbolic representation of the KM best response is hard, even when the strategy set is the hypercube $\{0, 1\}^d$.

\begin{theorem}[Hardness of computing the KM symbolic $\eps$-proper best response on the hypercube]
    \label{theorem:hardness-KM}
Given a vector $\vu \in \R^d$, define  a {\em KM best response} on the hypercube to be any vector\footnote{Note that there can be more than one such vector, due to tiebreaking.} $\vx \in [0, 1]^d$ of the form
    \begin{align*}
        \vx = \sum_{i=0}^{2^d-1} \eps^i \vs_i  
    \end{align*}
    where $\vs_0 \succeq  \dots \succeq \vs_{2^d-1}$ is the list of vertices of the hypercube $\{0, 1\}^d$, sorted in descending order of utility, and $\eps$ is symbolic. The following promise problem is $\sharpP$-hard under Turing reductions: given a utility vector $\vu$, compute the power of $\eps$ in the leading term in the polynomial $\vx^*[1]$, given the promise that this leading term is the same for all KM best responses $\vx^*$.
\end{theorem}
\begin{proof}
    We reduce from {\sc \#Knapsack}. We are given a vector $\vw \in \Z^{d-1}$ and a capacity $W \in \Z$, and our goal is to count how many $\vz \in \{0, 1\}^{d-1}$ there are with $\ip{\vw, \vz} \le W$.
    Consider the utility vector $\vu = (-(W+1/2), -\vec w) \in \R^{d}$.
    Then, by construction, we have that $\vs_1[1] = \vs_2[1] = \dots = \vs_k[1] = 0 \ne \vs_{k+1}[1] = 1$, where $k$ is the number of vectors $\vz$ for which $\ip{\vw, \vz} \le W$. Thus, the leading term of $\vx^*[1]$ is exactly $\eps^k$, and computing $\eps^k$ exactly solves the  {\sc \#Knapsack} instance.
\end{proof}

However, the KM best response is only one way to satisfy \Cref{def:eps-proper-br}, and one might ask whether there are others. Here, we give strong evidence to the contrary, in the form of showing that computing a proper equilibrium---that is, the limit point, not to mention a symbolic $\eps$-proper equilibrium---of even a two-player polytope potential game is $\NP$-hard.

\begin{theorem}[Hardness of computation]\label{th:nfpe-hard}
    The following promise decision problem is $\NP$-hard (and $\coNP$-hard, because it is trivially equivalent to its own negation): given a two-player identical-interest polytope game where P1's strategy set $\cX$ is given by explicit linear constraints and P2's pure strategy set is $\{0, 1\}$, decide whether 
    \begin{itemize}
        \item P2 deterministically plays $0$ in every normal-form proper equilibrium, or
        \item P2 deterministically plays $1$ in every normal-form proper equilibrium
    \end{itemize}
    given the promise that one of the two statements is true.
\end{theorem}
We dedicate the remainder of the subsection to the proof of the above theorem. We reduce from {\sc Partition}. We are given a vector $\vw \in \Z^d$, and our goal is to decide whether there is $\vz \in \{-1, 1\}^d$ with $\ip{\vw, \vz} = 0$. Define $\cX$ as 
\begin{align*}
    \cX &= \qty{ (t, \vz) \in [0, 1] \times \R^d : \norm{\vz}_\infty \le 1-t, \ip{\vw, \vz} \le 1-t}
    \\&= \conv\qty(\{(1, \vec 0)\} \cup ( \{0\} \times \cZ)) \qq{where} \cZ = \{ \vz \in [-1, 1]^d : \ip{\vw, \vz} \le 1\}.
\end{align*}
Let the utility function $u : \cX \times [0, 1] \to \R$ of both players be 
\begin{align*}
    u(\vx = (t, \vz), y) = \underbrace{-\frac12 t + \ip{\vw, \vz}}_{=: \tilde u (t, \vz)} + \underbrace{\frac{y}{3(1+\norm{\vw}_1)} \qty(1 - \ip{\vw, \vz} - 2t)}_{=: \delta(t, \vz, y)}.
\end{align*}
We will use the shorthand $\vx = (t, \vz)$ freely, and we will use $\bot = (1, \vec 0) \in \cX$. Intuitively, the rest of the proof proceeds as follows. The term $\delta(\cdot)$ is small and can be ignored when thinking about P1's incentive. The ordering of P1's pure strategies in decreasing order of utility is as follows: first, all those pure strategies $(0, \vz)$ where $\ip{\vw, \vz} \ge 1$; second, all those pure strategies $(0, \vz)$ where $\ip{\vw, \vz} = 0$; third, $\bot$; and finally, all other pure strategies, namely those of the form $(0, \vz)$ with $\ip{\vw, \vz} \le -1$. The second category are the solutions to the partition problem. Moreover, $t$ is precisely the probability of playing $\bot$, and in any $\eps$-proper best response, for sufficiently small $\eps$, the probability of each category is much smaller than the probability of the previous category. Thus, by comparing $1 - \ip{\vw, \vz}$ to $t$, we can understand whether the second category is empty. And P2's incentive is essentially to test which of these two quantities is bigger, so we can use P2's strategy in any proper equilibrium to determine whether a partition exists.
\begin{lemma}
    Let $\vs, \vs' \in \cS$. If $\tilde u(\vs) > \tilde u(\vs')$, then $u(\vs, y) > u(\vs', y)$ for all $y$.
\end{lemma}
\begin{proof}
    $\tilde u$ can only take on half-integral values when $t$ and $\vz$ are integral, and $0 \le \delta(t, \vz, y) \le 1/3$ for all $(t, \vz, \vy)$, so adding $\delta$ cannot affect the ordering of the $\vx$s.
\end{proof}
Let $(t^*, \vz^*, y^*)$ be an $\eps$-proper equilibrium. Order the vertices of $\cX$ as $\cS = \{ \vs_1, \dots, \vs_N \}$ where $\vs_i \succeq \vs_{i+1}$ for all $i$, and $\succeq$ is with respect to P1's utility $u(\cdot, y^*)$. Note that, by the previous lemma, this ordering also applies to $\tilde u$, that is, we have $\tilde u(\vs_i) \ge \tilde u(\vs_{i+1})$ for all $i$. We will use this fact freely. By definition of $\eps$-proper best response, there is a full-support distribution $\lambda \in \Delta(N)$ such that 
   $
        \vx^* = \E_{i \sim \lambda} \vs_i
   $
    and $\lambda_j \le \eps \lambda_i$ whenever $\vs_j \prec \vs_i$. Let $i^*$ be the smallest index for which $\tilde u(\vs_{i^*}) < 1$. Notice that, for all $i < i^*$, we must have $t_i = 0$ and $\ip{\vw, \vz_i} = 1$.

\begin{lemma}
    If the {\sc Partition} instance has a solution, then $y^* \ge 1 - \eps$.
\end{lemma}
\begin{proof}
    In this case, we must have $\tilde u(\vs_{i^*}) = 0$, and in fact $\vz_{i^*}$ must be a solution to the partition instance. Thus, $1 - \ip{\vw, \vz^*} \ge \lambda_{i^*}$. Moreover, since $\tilde u(\bot) = -1/2 < 0$, the weight $\lambda$ places on $\bot$ must be at most $\eps \cdot \lambda_{i^*}$, so $t \le \eps \cdot \lambda_{i^*}$. Thus $1 - \ip{\vw, \vz^*} - 2t^* > 0$, and in order for $y^*$ to be a best response, we must have $y^* \ge 1 - \eps$.
\end{proof}
\begin{lemma}
    If the {\sc Partition} instance has no solution, then $y^* \le \eps$.
\end{lemma}
\begin{proof}
    In this case, we have $\vs_{i^*} = \bot$. Moreover, the very next strategy $\vs_{i^*+1}$ must be a vertex of $\cZ$ and thus must have $\tilde u(\vs_{i^*+1}) \le -1$, because $\tilde u(\bot) = -1/2$ and $\bot$ is the only strategy with a non-integral value of $\tilde u$. Thus, $\lambda_j \le \eps \lambda_{i^*}$ for every $j > i^*$, and therefore we have
    \begin{align*}
        1 - \ip{\vw, \vz^*} \le \lambda_{i^*} + \sum_{j > i} \lambda_j \cdot \norm{\vw}_1 \le \lambda_{i^*} \qty(1 + \eps N \norm{\vw}_1) < 2 \lambda_{i^*} = t
    \end{align*}
    for sufficiently small $\eps$. Thus $1 - \ip{\vw, \vz^*} - 2t^* < 0$, and in order for $y^*$ to be a best response, we must have $y^* \le \eps$.
\end{proof}
Thus, passing to the limit, if the {\sc Partition} instance has a solution, then every proper equilibrium has $y=1$; if the {\sc Partition} instance has no solution, then every proper equilibrium has $y=0$. This completes the proof. 

\begin{remark}
    The above result implies that, if a normal-form proper equilibrium can be computed in polynomial time, then $\P = \NP$. It does {\em not} imply that computing a normal-form proper equilibrium is \FNP-hard, because the proof does not allow the extraction of the partition in the case where one exists, which would be required for \FNP-hardness. We leave closing this gap to future research.
\end{remark}

In the above proof, P1's order of pure strategies does not depend at all on P2's strategy. Thus, if $\vx^* = (t, \vz)$ is any optimizer of $\ip{\vw, \vz}$, then one of the two profiles $(\vx^*, 0)$ and $(\vx^*, 1)$ is a proper equilibrium. It thus follows that:
\begin{corollary}
    Checking whether a given profile is a proper equilibrium is $\NP$-hard, even for two-player identical-interest polytope games. 
\end{corollary}

\section{Efficient best responses for extensive-form games}\label{sec:efg br}
In this section, we will show the surprising result that, despite the hardness for general polytope games, proper best responses {\em are} efficiently computable in {\em extensive-form games}. We will then use the ability to find best responses to show that computing a normal-form proper equilibrium of an extensive-form game is no harder than computing a Nash equilibrium.

Let $\cS \subset \{0, 1\}^\Sigma$ be a sequence-form pure strategy set with  dimension $|\Sigma|=d$ and size $|\cS| = N$, and $\cX = \conv \cS$. Fix a utility vector $\vu \in \R^\Sigma$. For each sequence $\sigma \in \Sigma$, let $u^*(\sigma)$ be the optimal value that can be attained by any strategy that plays to sequence $\sigma$ with probability $1$. That is, $u^*(\sigma)$ is the optimal value of the linear program
\begin{align}
    \max\quad \ip{\vx, \vu} \qq{s.t.} \vx \in \cX, \quad \vx[\sigma] = 1. \label{eq:conditional br}
\end{align}
Call a strategy $\vx \in \cX$ an {\em $\eps$-sequentially proper best response} if
\begin{enumerate}
    \item $\vx$ is fully mixed, \ie, $\vx[\sigma] > 0$ for every sequence $\sigma$, and
    \item $\vx[\sigma] \le \eps \cdot \vx[\sigma']$ whenever $u^*(\sigma) < u^*(\sigma')$.
\end{enumerate}
We demand that the latter condition hold for {\em all} pairs $\sigma, \sigma'$---not just, for example, pairs where $\sigma \succeq \sigma'$. Similarly, call a strategy profile of a game a {\em $\eps$-sequentially proper equilibrium} if each player is playing an $\eps$-sequentially proper best response, and a {\em sequentially proper equilibrium} if it is a limit point of $\eps$-sequentially proper equilibria as $\eps \to 0$.  

\begin{lemma}\label{lem:proper br prob lower bound}
    For every $\vu$ and every $\eps \in (0, 2/3]$, there exists an $\eps$-sequentially proper best response $\vx$ such that $\vx[\sigma] \ge (\eps/2)^{d}$ for every sequence $\sigma$. Moreover, there is an efficient algorithm for finding a pure symbolic $\eps$-proper best response.
\end{lemma}
\begin{proof}
    Number the sequences $\Sigma = \{ \sigma_0, \dots, \sigma_{d-1}\}$ in descending order of $u^*(\sigma)$, tiebreaking in favor of having higher-up sequences with lower numbers. Let $k(\sigma)$ be the index of sequence $\sigma$ in this ordering. In particular, $\sigma_0 = \Root$ and $\sigma_{d-1}$ is a leaf. Moreover, for each decision point $j$ let $a^*(j)$ be a best-responding action at $j$, breaking ties arbitrarily. Then consider the strategy recursively defined by $\vx^*[\Root] = 1$, and
    \begin{align*}
        \vx^*[ja] := \begin{cases}
            \vx^*[p_j] - \sum_{a' \in A_j \setminus a} \hat\eps^{k(ja')}&\qif a = a^*(j) \\
            \hat\eps^{k(ja)} &\qq{otherwise}
        \end{cases}
    \end{align*}
    where $\hat\eps$ will be picked later. This can be computed efficiently. Every entry of $\vx^*$ has the form $\vx^*[\sigma] = \hat\eps^k - \sum_{k' \in S} \hat\eps^{k'}$ where $S$ contains only natural numbers larger than $k$, and $k$ is the index of some sequence $\sigma'$ with $u^*(\sigma') = u^*(\sigma')$. Moreover,
    \begin{align*}
        \hat\eps^k - \sum_{k' \in S} \hat\eps^{k'} \ge \hat\eps^k - \sum_{k'=k+1}^\infty \hat\eps^{k'} = \hat\eps^k - \frac{\hat\eps^{k+1}}{1 - \hat\eps} \ge \frac12 \hat\eps^k > 0
    \end{align*}
    for $\hat\eps \le 1/3$, so $\vx^*$ is indeed a valid strategy. Therefore, if $u^*(\sigma) < u^*(\sigma')$, then we have $\vx^*[\sigma] \le 2\hat\eps \cdot \vx^*[\sigma']$, and taking $\hat\eps = \min\{ \eps/2, 1/3\}$ completes the proof.
\end{proof}

Moreover, the problem of finding an $\eps$-sequentially proper best response can be expressed as a feasibility problem with linear constraints, given the values $\{u^*(\sigma)\}_{\sigma \in \Sigma}$:
\begin{align}
\begin{aligned}
    \qq{find} \vx \in \cX &\qq{s.t.} \\\vx[\sigma] \ge (\eps/2)^{d} &\qq{for all $\sigma$,} \\\vx[\sigma] \le \eps \cdot \vx[\sigma'] &\qq{for all $\sigma, \sigma'$ with $u^*(\sigma) < u^*(\sigma')$}
\end{aligned}\label{eq:sf proper br lp}
\end{align}
This LP formulation, whose feasibility is guaranteed by \Cref{lem:proper br prob lower bound}, will be useful later on. Moreover, since the lower bound $\vx[\sigma] \ge (\eps/2)^{d}$ can be enforced without loss of generality, we will enforce it; that is, when we say ``$\eps$-sequentially proper best response'' from now on, we will mean a solution to \eqref{eq:sf proper br lp}.

The main purpose of defining $\eps$-sequentially proper best responses is that they are computationally tractable, while being {\em equivalent} to the usual notion of $\eps$-proper best responses, in the following formal sense.
\begin{theorem}\label{th:equivalence}
    Let $\eps > 0$, and $\hat\eps = \eps^{N+1}/(8N^2d)$. Then
    \begin{enumerate}
        \item every $\hat\eps$-proper best response is an $\eps$-sequentially proper best response, and
        \item every $\hat\eps$-sequentially proper best response is an $\eps$-proper best response.
    \end{enumerate}
    Hence, the sequentially proper equilibria coincide with the normal-form proper equilibria.
\end{theorem}
Note that the bound in this theorem is much worse than the bound in \Cref{lem:proper br prob lower bound}---the latter is only singly-exponential in $d$, whereas this theorem's bound in doubly-exponential. The fact that \Cref{lem:proper br prob lower bound} has a singly-exponential bound will be crucial later on. 

The remainder of the section is dedicated to the proof of this theorem.
For notation, let $\cS^\sigma = \{ \vx \in \cS:  \vx[\sigma] = 1\}$, and $\cX^\sigma = \conv \cS^\sigma$. For any strategy $\vx$ and decision point (or sequence) $j$, define $\vx_{\succeq j}$ to be the sub-vector of $\vx$ indexed on all sequences that are successors of $j$. Notation such as $\vx_{\not\succeq j}$, $\vx_{\succ \sigma}$, {\em etc.} is defined analogously.

    For the first statement, let $\lambda \in \Delta(\cS)$ satisfy the conditions of an $\hat\eps$-proper best response. We need to show that $\vx^* := \E_{\vx'\sim\lambda}\vx'$ is an $\eps$-proper best response. Indeed, consider any two sequences $\sigma, \sigma'$ with $u^*(\sigma) < u^*(\sigma')$. Then, by definition, there is some $\vs^{\sigma'} \in \cS^\sigma$ such that 1) $\vs^{\sigma'}[\sigma'] = 1$, and 2) for every $\vs^\sigma \in \cS^\sigma$, we have $\vs^\sigma \prec \vs^{\sigma'}$. By the $\hat\eps$-proper best response condition, we thus have $\lambda(\vs^\sigma) \le \hat\eps \cdot \lambda(\vs^{\sigma'})$ for every such $\vs^\sigma$. Thus, we have
    \begin{align*}
        \vx^*[\sigma] = \sum_{\vs^\sigma \in \cS^\sigma} \lambda(\vs^\sigma) \le \hat\eps \cdot |\cS^\sigma|  \cdot \lambda(\vs^{\sigma'}) \le \hat\eps \cdot |\cS| \cdot \vx^*[\sigma'] \le \eps \cdot \vx^*[\sigma']
    \end{align*}
    as desired.

For the reverse direction, we will explicitly specify an algorithm for constructing, from an $\hat\eps$-sequentially proper best response $\vx^*$, a distribution $\lambda$ with the same expectation that is an $\eps$-proper best response. We start by constructing the behavioral form $\lambda$ of $\vx^*$. Formally, $\lambda \in \Delta(\cS)$ is the distribution
\begin{align}
    \lambda(\vs) = \prod_{ja : \vs[ja]=1} \beta_{\vx^*}(a|j) \qq{where} \beta_{\vx^*}(a|j) := \frac{\vx^*[ja]}{\vx^*[p_j]}.\label{eq:behavioral}
\end{align}
For each sequence $\sigma$, define $\vs^\sigma \in \cS^\sigma$ to be the highest-probability (under $\lambda$) pure strategy among those pure strategies that play to sequence $\sigma$. Assume there are no ties without loss of generality; otherwise, consider $\vu$ to be perturbed infinitesimally to remove ties for purposes of selecting the $\vx^\sigma$s; this does not change the remainder of the proof.

\begin{lemma}\label{lem:behavioral properties}
    For every $\sigma$, 
    \begin{enumerate}
        \item $\vx^*[\sigma] \ge \lambda(\vs^\sigma) \ge \vx^*[\sigma] / N$, and
        \item if $\vs \in \cS^\sigma$ and $\ip{\vs, \vu} < u^*(\sigma)$, then $\lambda(\vs) \le \hat\eps \cdot \vx^*[\sigma]$.
        \item $\ip{\vs^\sigma, \vu} = u^*(\sigma)$.
        \item if $\sigma = ja$ is not the root sequence, then $\vs^\sigma_{\not\succeq j}[\sigma'] = \vs^{p_j}_{\not\succeq j}[\sigma']$; and if $\vs^{p_j}[\sigma] = 1$ then $\vs^\sigma = \vs^{p_j}$.
    \end{enumerate}
\end{lemma}
\begin{proof}
    The first claim is immediate from the fact that $\E_{\vs\sim\lambda}\vs[\sigma] = \vx^*[\sigma]$ by definition of the sequence form. 
    
    For the second, if $\vs$ were not a best response under this constraint, then there is some decision point $j\not\preceq\sigma$ at which $\vs$ plays suboptimal action $a$. But then $\beta_{\vx^*}(a|j) \le \hat\eps$, and since $j \not\preceq \sigma$ it follows that $\lambda(\vs) \le \vx^*[\sigma] \cdot \beta_{\vx^*}(a|j) \le \hat\eps \cdot \vx^*[\sigma]$. 

    The third claim follows immediately from the first two, by taking $\hat\eps < 1/N^2$.

    For the third claim, notice that \eqref{eq:behavioral}, for $\vs \in \cS^{p_j}$, is a product of two independent terms:
    \begin{align*}
        \lambda(\vs) = \qty\bigg(\prod_{\substack{j'a' \succeq j,\\ \vs[j'a']=1}} \beta_{\vx^*}(a'|j')) \qty\bigg(\prod_{\substack{j'a' \not\succeq j,\\ \vs[j'a']=1}} \beta_{\vx^*}(a'|j'))
    \end{align*}
    where, since $\vs[p_j] = 1$, the two terms can be optimized independently. Therefore, at optimality, the values $\vs[\sigma']$ for sequences $\sigma' \not\succeq j$ (\ie, $\sigma'$ arising in the second term) are unaffected by whether the additional constraint $\vs[ja] = 1$ exists. Finally, if $\vs^{p_j}[\sigma] = 1$, then $\vs^{p_j} \in \cS^{\sigma}$, so adding the extra constraint $\vs[\sigma] = 1$ does not change this optimal solution.
\end{proof}
Therefore, in particular, the first claim in the lemma shows that the vertices $\vs^\sigma$ already obey the $\eps$-proper constraint under distribution $\lambda$, for $\hat\eps \le \eps/N$. However, it is not necessarily the case that {\em all} vertices obey the $\eps$-proper constraint. The remainder of the proof will go as follows: we will show how to {\em adjust} the probabilities $\lambda(\vs)$, while maintaining the fact that the expectation is still $\vx^*$, in such a way that the resulting distribution obeys the $\eps$-proper constraint. Intuitively, this is possible because the vertices $\vs \in \cS$ have affine dependencies---so, one can write $\vs$ as an affine combination of other vertices, and add or remove mass from this group of vertices together, so that $\lambda(\vs)$ changes while  $\E_{\vs\sim\lambda}\vs$ does not.

\begin{lemma}
    Let $\vs \in \cS$ and $B_\vs := \{ \vs^\sigma : \sigma \in \supp \vs\}$. Then $\vs$ is an affine combination of elements of $B_{\vs}$ with coefficients of magnitude at most $d$. In symbols,\label{lem:basis}
    \begin{align*}
        \vs = \sum_{\vs^\sigma \in B_\vs} c(\vs, \sigma) \cdot \vs^\sigma
    \end{align*}
    where $|c(\vs, \sigma)| \le d$ and $\sum_{\vs^\sigma \in B_\vs} c(\vs, \sigma) = 1$.
\end{lemma}
\begin{proof}
Let $\vs$ be an arbitrary vertex. We will describe an algorithm that converts $\vs$ into $\vs^\Root$ by a series of iterative steps that only move in $\aff B_{\succeq \vs}$. 

The algorithm is the following. For each decision point $j$ in bottom-up (leaves-first) order, if $\vs[p_j] = 1$, then perform the operation
\begin{align}
    \vs \gets \vs + \vs^{p_j} - \vs^{ja} \label{eq:iterative spanning set proof}
\end{align}
where $a$ is the action played by $\vs$ at $j$. We make the following inductive claims about this algorithm: when the operation \eqref{eq:iterative spanning set proof} is executed at decision point $j$:
\begin{enumerate}
    \item before the operation, $\vs \preceq \vs^{ja} \preceq \vs^{p_j}$; therefore, $\ip{\vs, \vu}$ only (weakly) increases, and all changes to $\vs$ are moves within the affine hull of $B_{\vs}$. 

    {\em Proof.} Both inequalities follow from the earlier lemma, which showed that $\vs^\sigma$ is a best response under the constraint $\vs^\sigma[\sigma] = 1$.
    \item after the operation,  $\vs_{\succeq j} = \vs^{p_j}_{\succeq j}$. 
    
    {\em Proof.} When the operation is executed at $j$, we have $\vs_{\succeq j} = \vs^{ja}_{\succeq j}$, where $a$ is the action played by $\vs$, by inductive hypothesis, and $\vs^{p_j}_{\not\succeq j}$ agrees with $\vs^{ja}_{\not\succeq j}$ by construction of the $\vs^\sigma$s. 
\end{enumerate}
But then, after the algorithm is finished, we have $\vs = \vs^\Root$. Moreover, only at most $d$ steps were taken, each of which involved the addition and subtraction of a single $\vs^\sigma$; therefore, the coefficients in the affine combination lie in the range $[-d, d]$. 
\end{proof} 

Let $\vs \in \cS$. By the previous lemma, $\vs$ is an affine combination of $B_\vs := \{ \vs^\sigma : \sigma \in \supp \vs\}$, and by construction we have $\vs \preceq \min B_\vs$.

\begin{lemma}\label{lem:movement ub}
    If $\vs \prec \min B_\vs$, then $\lambda(\vs) \le \min_{\vs^\sigma \in B_\vs}\lambda(\vs^\sigma) / (2Nd) =: \bar\lambda(\vs)$.
\end{lemma}
\begin{proof}
    By \Cref{lem:behavioral properties}, for $\hat\eps < 1/N$, the $\vs^\sigma$s are correctly sorted in the sense that if $\vs^\sigma \prec \vs^{\sigma'}$ then $\lambda(\vs^\sigma ) < \lambda(\vs^{\sigma'})$. Thus, the minimizer of $\lambda(\vs^\sigma)$ is also a minimizer of utility. Let $\vs^\sigma = \argmin_{\vs^\sigma \in B_\vs}\lambda(\vs^\sigma)$ be this shared minimizer. Then by the first two items of \Cref{lem:behavioral properties}, for $\hat\eps \le 1/(2N^2d)$, we have $\lambda(\vs) \le \hat\eps \cdot \vx^*[\sigma] \le N \hat\eps \cdot \lambda(\vs^\sigma) \le \lambda(\vs^\sigma)/(2Nd)$, as desired.
\end{proof}

Now let $\cS^*$ be the set of all vertices for which $\lambda(\vs) > \bar\lambda(\vs)$.
\begin{align*}
    \lambda^* =\lambda + \sum_{\vs \notin\cS^*} \alpha(\vs) \qty(\ind(\vs) - \sum_{\vs^\sigma \in B_\vs} c(\vs, \sigma) \cdot \ind(\vs^\sigma))
\end{align*}
where $|\alpha(\vs)| \le \bar\lambda(\vs)$ for all $\vs$, $B = \{ \vs^\sigma : \sigma \in \Sigma\}$, and $\ind(\vs)$ is the deterministic distribution on $\vs$. Then by the previous \Cref{lem:basis},  $\E_{\vx\sim\lambda^*}\vx = \E_{\vx\sim\lambda}\vx = \vx^*$, and by collecting like terms, we have $\lambda(\vs)/2 \le \lambda^*(\vs) \le 2\lambda^*(\vs)$ for all $\vs \in \cS^*$. But since the $\alpha(\vs)$s can be set freely, by \Cref{lem:movement ub}, it follows that for $\vs \notin \cS^*$, the probability $\lambda^*(\vs)$ can be set to {\em any} value $t \in [0, \bar\lambda(\vs)]$, by picking $\alpha(\vs) = t - \lambda(\vs)$.

\begin{lemma}
    For all pairs of vertices $\vs, \vs' \in \cS^*$, the distribution $\lambda^*$ obeys the $\eps$-proper best response condition regardless of the setting of the $\alpha$. That is, if $\vs \prec \vs'$ then $\lambda(\vs) \le \eps \cdot \lambda(\vs')$.
\end{lemma} 
\begin{proof}
    For any $\vs \in \cS^*$, let $\sigma$ be the minimizing sequence in the proof of \Cref{lem:movement ub} for that $\vs$. Then by that lemma, we have $\ip{\vu, \vs} = \ip{\vu, \vs^\sigma}$, and from \Cref{lem:behavioral properties}, we have
    \begin{align*}
        \frac{\vx^*[\sigma]}{4N^2d} \le \frac{\lambda(\vs^\sigma)}{4Nd}\le \frac{\lambda(\vs)}{2}\le \lambda^*(\vs) \le 2 \lambda(\vs) \le 2\lambda(\vs^\sigma) \le 2\vx^*[\sigma]
    \end{align*}
    The lemma then follows from setting $\hat\eps \le \eps/(8N^2 d)$ and using the $\hat\eps$-sequentially proper best response condition.
\end{proof}

Intuitively, we imagine each $\lambda^*(\vs)$ as belonging to a ``slot'' of (multiplicative) width $8N^2d$. If $\hat\eps$ is set small enough that these slots are $\eps$ apart from each other (\ie, $\hat\eps \le \eps/(8N^2d)$), and they are sorted correctly (which is guaranteed by the above lemma), then this is enough to guarantee the $\eps$-proper best response condition.

It thus remains only to select $\lambda^*(\vs) \in (0, \bar\lambda(\vs)]$ for each $\vs \notin \cS^*$ in such a way that these also obey the $\eps$-proper best response condition. For each such $\vs$, let $\sigma$ be the minimizer in \Cref{lem:movement ub}, and let
\begin{align*}
    \sigma^* = \argmin_{\sigma' : \vs^{\sigma'} \succeq \vs} \lambda(\vs^{\sigma'}).
\end{align*}
Finally, let $N(\vs)$ be the rank of vertex $\vs$ among {\em all} vertices $\vs \in \cS$ with respect to utility, that is, $N(\vs) = |\{ \vs' : \vs \prec \vs'\}|$. We set $$\lambda^*(\vs) = \frac{\vx^*[\sigma^*]}{4N^2d} \cdot \eps^{-N(\vs)}.$$
Then $\lambda^*(\vs) \le \bar\lambda(\vs)$ because $\vs^{\sigma'} \succeq \vs$ for every $\vs^{\sigma'} \in B_\vs$, so this setting is valid. Moreover, if we take $\hat\eps \le \eps^{N+1}/(8N^2d)$, then enough space is created between each slot to fit up to $N$ distinct values of $\lambda^*(\vs)$ in sorted order, which is a sufficient for the $\eps$-proper best response condition also holds for $\vs, \vs'$ where one or both are outside $\cS^*$. Hence, the proof is complete.
\section{Computing proper equilibria of extensive-form games}\label{sec:efg-proper}
By \Cref{th:equivalence}, computing a normal-form proper equilibrium of an extensive-form game is equivalent to computing a sequentially proper equilibrium. Here, we show how to compute such an equilibrium. To do so, we first require some setup.

We will need the following result, which is a special case of the framework of \citet{FilosRatsikas21,FilosRatsikas24:PPAD}:
\begin{definition}[Special case of \citealp{FilosRatsikas21}]
    A {\em conditional linear feasibility gate} ({\em CLF gate}) is parameterized by integers $n, m > 0$, a matrix $\mA \in \R^{m \times n}$, and a radius $R$. It takes as input vectors $\vb, \vh \in \R^m$ , and outputs a solution to the linear system\footnote{ As discussed by \citet{FilosRatsikas21}, if there are multiple solutions, any solution is a valid output; and the matrix $\mA$ must not be input-dependent.} 
    \begin{align}\label{eq:linopt-lp}
        \qq{find} \quad \vx \in [-R, R]^n \qq{s.t.} [\vec h]_+ \circ (\mA \vx - \vb) \le 0.
    \end{align}
\end{definition}
That is, if $h_i > 0$ then the constraint $\vec a_i^\top \vx \le b_i$ is enforced; otherwise it is not enforced. Unconditional constraints $\vec a_i^\top \vx \le b_i$ can easily be incorporated by setting $h_i = 1$. 
\begin{theorem}[\citealp{FilosRatsikas21}]\label{th:linopt}
    CLF gates can be used freely as part of the construction of an arithmetic circuit for $\FIXP$ and $\PPAD$ proofs. More precisely, for any CLF gate parameterized by $\mA \in \R^{m \times n}$ and $R$, we can construct in time $\poly(m, n, \size(A), \size(R))$ an arithmetic circuit with gates $\{+, -, *c, \max, \min\}$ (where $*c$ is multiplication by a rational constant) $F : \R^m \times \R^m \times [0, 1]^t \to \R^n \times [0, 1]^t$, such that the following holds: for all $\vb, \vh \in \R^n$, $\vx \in \R^n$, and $\vec\alpha \in [0,1]^t$, if
    \begin{enumerate}[nosep]
        \item the linear system \eqref{eq:linopt-lp} is feasible,
        \item $F(\vb, \vh, \vec\alpha) = (\vx, \vec\alpha)$,
    \end{enumerate}
    then $\vx$ is a solution to \eqref{eq:linopt-lp}. (If \eqref{eq:linopt-lp} is infeasible, then arbitrary output is possible.)
\end{theorem}

We will also need the following result of \citet{Farina17:Extensive}, which will allow us to set $\hat\eps$ appropriately. Let $\mathcal{P}(\eps)$ be a family of linear complementarity problems whose data (\ie, constraint matrices and bounds) are expressed as polynomials of degree at most $d$ in a single variable $\eps$. We call a point $\vx^*$ a {\em limit solution} of $\mathcal P$ if there are sequences $\eps^{(k)} \to 0$, $\vx^{(k)} \to \vx^*$ such that each $\vx^{(k)}$ is a solution to $\mathcal P(\eps^{(k)})$. 
\begin{lemma}[\citealp{Farina17:Extensive}]
    \label{lemma:NPP}
     There exists some $\eps^* > 0$, expressible in a number of bits polynomial in $d$ and the representation size of $\mathcal{P}$, such that every complementary basis for $\mathcal{P}(\eps^*)$ remains a complementary basis for $\mathcal{P}(\eps)$ for all $0 < \eps \le \eps^*$. Thus, in particular, a limit solution can be found by solving $\mathcal P(\eps)$ for $\eps \le \eps^*$, and then using the same basis to derive a solution to $\mathcal P(0)$. 
\end{lemma}

\begin{lemma}\label{lem:efg nfpe circuit}
    Let $\Gamma$ be an extensive-form game with $n$ players, and for each player $i \in [n]$ let $\cX_i \subset \R^{d_i}$ be that player's sequence-form strategy set, and $d = \sum_i d_i$. Then there is a $\poly(\size(\Gamma))$-time algorithm that constructs an arithmetic circuit $F : [0, 1] \times \cQ \to \cQ$ with gates $\{+, -, \times, \max, \min\}$, where $\cQ = [0, 1]^{d} \times [0, 1]^t$, such that the following holds: For each $\eps > 0$, let $\vx^{(\eps)} \in [0, 1]^{d}$ and $\vec\alpha^{(\eps)} \in [0, 1]^t$ be such that $$F(\eps, \vx^{(\eps)}, \vec\alpha^{(\eps)}) = (\vx^{(\eps)}, \vec\alpha^{(\eps)}).$$
    Then $\vx^{(\eps)}$ is an $\eps$-sequentially proper equilibrium, and therefore every limit point of $\{\vx^{(\eps)}\}_{\eps \to 0^+}$ is a normal-form proper equilibrium of $\Gamma$. Moreover, when $\Gamma$ is a polymatrix game, $F$ is piecewise linear if $\eps$ is treated as a constant.
\end{lemma}

For the avoidance of doubt, an extensive-form game is {\em polymatrix} if $u_i(\vx_i, \vx_{-i})$ is linear in both $\vx_i$ and $\vx_{-i}$. In particular, this holds if each root-to-leaf path contains only decision points for at most two players. Two-player games are obviously polymatrix, as are normal-form polymatrix games. 
\begin{proof}
    By \Cref{th:linopt}, it is enough to construct a circuit $\tilde F : (\eps, \vx) \mapsto \vx$ that includes CLF gates. The circuit $\tilde F$ performs the following operations, for each player $i$:
    
    \begin{enumerate}
        \item Compute the utility vector $\vu_i := u_i(\cdot, \vx_{-i}) \in \R^{d_i}$. For polymatrix games, this is linear in the input; for general $n$-player games, this is a polynomial of degree $n-1$ in the input. 
        \item For each sequence $\sigma_i \in \Sigma_i$, compute $u^*(\sigma_i)$, defined in \eqref{eq:conditional br}. This can be done using an iterative bottom-up pass through the sequence-form tree of player $i$, using max and sum gates.
        \item Compute an $\eps$-sequentially proper  best response $\vx_i'$ by solving the CLF system \eqref{eq:sf proper br lp}. The constant $(\eps/2)^d$, required by the LP formulation, can be created by repeated multiplication by $\eps/2$, which is legal because $\eps$ is treated as a constant.
    \end{enumerate} 
    Finally, $\tilde F$ outputs the profile $\vx'$. By construction, if $\vx = \vx'$ then $\vx_i$ is an $\eps$-sequentially proper  best response for every $i$.
\end{proof}

It remains to show that one can set $\eps$ small enough that $\vx^{(\eps)}$ is close to a limit point of $\{\vx^{(\eps)}\}_{\eps \to 0^+}$. This requires separate proofs for the polymatrix case and the $\FIXP$ case. 

\begin{lemma}\label{lem:efg polymatrix lcp}
    Let $F : [0, 1] \times [0, 1]^m \to [0, 1]^m$ be a piecewise linear arithmetic circuit. Then for every $\eps > 0$ there is a linear complementarity problem of size $\poly(\size(F))$ whose parameters are linear functions of $\eps$ and whose solutions correspond to points $\vz \in [0, 1]^m$ for which $F(\eps, \vz) = \vz$.
\end{lemma}
\begin{proof}
    It suffices to show how to implement the $+,-, \times$ (by a constant), and $\max$ gates using an LCP. The $+,-$, and $\times$ gates are straightforward linear equality constraints. A $\max$ gate $a = \max\{b, c\}$ is equivalent to the constraint system $a \ge b, a \ge c, (a-b) \cdot (a-c) = 0$.
\end{proof}

For completeness, we include an LCP for extensive-form polymatrix games.
For each player $i$ let $\mA_i$ be the matrix for which $\vx_i^\top \mA_i \vx_{-i} = u(\vx_i, \vx_{-i})$. To create the value $\eps^d$, we create $d$ variables $\eps_1, \dots, \eps_d$ and set $\eps_1 = \eps$, and add the linear constraint $\eps_t = \eps \cdot \eps_{t-1}$ for $t > 1$. Then $\eps_d = \eps^d$. Then, for each player $i$, we add the following constraints, using the constraint system in the previous lemma to implement the max constraints.
$$
    \begin{aligned}
    \vx_i &\in \cX_i & \\
    \vx_i &\ge 2^{-d}\eps_d \vec 1 \\
            \vu_i &= \mA_i \vx_{-i} &\\
            u^*_i[\sigma, \sigma'] &= \sum_{j \in C_{\sigma'}} u^*_i[\sigma, j] \quad && \forall \sigma, \sigma' \in \Sigma_i \\
            u^*_i[\sigma, j] &= u^*_i[\sigma, ja] && \forall \sigma, ja \text{ s.t. } ja \preceq \sigma \\
            u^*_i[\sigma, j] &= \max_{a \in \cA_j} u^*_i[\sigma, ja] && \forall \sigma, j \text{ s.t. } j \not\preceq \sigma \\
            [\vx_i[\sigma] - \eps \cdot \vx_i[\sigma']]_+ \cdot [u^*_i[\sigma', \Root] - u^*_i[\sigma, \Root]]_+ &= 0&& \forall \sigma, \sigma' \in \Sigma_i
    \end{aligned}
$$
The variable $u_i^*[\sigma, \Root]$ represents the value $u_i^*[\sigma]$ introduced in \Cref{sec:efg br}, and is computed here using a bottom-up pass. The final constraint represents the conditional constraint in \eqref{eq:sf proper br lp}: if $u_i^*[\sigma'] - u_i^*[\sigma] > 0$ then $\vx_i[\sigma] - \eps \cdot \vx_i[\sigma'] \le 0$. Thus, these constraints encode the fact that $\vx_i$ is an $\eps$-sequentially proper best response for each $i$, as desired.

\begin{lemma}
    \label{lemma:FIXP-almostnear}
    Consider a bounded domain $D \subset \R^d$ expressed with linear constraints and a circuit $(\vx, \epsilon) \mapsto F_\epsilon(\vx) $ with a set of gates $\cG$, each of which performs either addition, multiplication, maximum, or minimum. For any $\delta > 0$ and sufficiently small $\epsilon, \epsilon' \geq \delta^{2^{O(d^3 |\cG| )}} $, any $\epsilon'$-almost fixed point of $F_\epsilon(\vx)$ is $\delta$-close to a limit point of fixed points of $F_{\epsilon''}(\cdot)$ when $\epsilon'' \to 0^+$.
\end{lemma}

\begin{proof}
For a given $\vx \in \R^d$, $\epsilon > 0$, and $\epsilon' > 0$, we define the formula $\textsc{AlmostFP}(\vx, \epsilon, \epsilon')$ that represents whether $\vx$ is an $\epsilon'$-almost fixed point of $F_\epsilon(\vx)$. We consider a set of auxiliary variables $\vy$ with existential quantifiers. For notational convenience, we write $\tvx = (\vx, \epsilon)$. We first add the clauses $\vy[i] = \tvx[i]$ for all $i$ that appear as inputs in the first level of the circuit. We also add the following clauses for all gates $G \in \cG$:
\begin{align*}
    \vy[o] = \vy[i_1] + \vy[i_2] \text{ for } G = (+, i_1, i_2, o);\\
    \vy[o] = \vy[i_1]*\vy[i_2] \text{ for } G = (*, i_1, i_2, o); \\
    \vy[o] = \zeta* \vy[i] \text{ for } G = (*, i, \zeta, o); \\
    (\vy[i_1] \geq \vy[i_2] \implies \vy[o] = \vy[i_1]) \land (\vy[i_2] \geq \vy[i_1] \implies \vy[o] = \vy[i_2]) \text{ for } G = (\max, i_1, i_2, o); \\
    (\vy[i_1] \geq \vy[i_2] \implies \vy[o] = \vy[i_2]) \land (\vy[i_2] \geq \vy[i_1] \implies \vy[o] = \vy[i_1]) \text{ for } G = (\min, i_1, i_2, o).
\end{align*}

We then add the clause $\| \tvx - \vy[out] \|^2_2 \leq (\epsilon')^2 $, where $\vy[out]$ collects all the output variables of $\vy$. Finally, we add clauses expressing $\vx \in D$, assumed to be linear constraints. What we want to prove is that for a fixed $\delta > 0$, and free variables $\epsilon, \epsilon'$,
\begin{equation*}
    \forall \vx \in \R^d \exists \vx' \in \R^d : (\epsilon, \epsilon' > 0) \land ( \lnot \textsc{AlmostFP}(\vx, \epsilon, \epsilon') \lor ( \textsc{RefinedFP}(\vx')) \land \| \vx - \vx' \|_2^2 \leq \delta^2 ),
\end{equation*}
where $\textsc{RefinedFP}(\vx') \defeq \forall \epsilon'' > 0 \exists \vx'' \in \R^d : \textsc{AlmostFP}(\vx'', \epsilon'', 0) \land ( \|\vx' - \vx''\|_2^2 \leq (\epsilon'')^2 )$. We denote by $\textsc{AlmostNearBound}_\delta(\epsilon, \epsilon')$ the above formula. Following the proof of~\citet[Lemma 4]{Etessami14:Complexity}, it follows that $\textsc{AlmostNearBound}_\delta(\epsilon, \epsilon')$ is satisfied for some $\epsilon, \epsilon' \geq \delta^{2^{O(d^3 |\cG|)}}$.
\end{proof}

Combining these results, we can now prove the main result of this section:
\begin{theorem}\label{th:main}
In extensive-form games, the problem of computing an exact normal-form proper equilibrium is in:
\begin{enumerate}[(1), noitemsep]
    \item $\FIXP_a$ for general games,
    \item $\PPAD$ for polymatrix games,
    \item $\PLS$ for potential games, and
    \item $\CLS$ for polymatrix potential games.
\end{enumerate}
\end{theorem}
\begin{proof}
    (1) follows by combining \Cref{lem:efg nfpe circuit,lemma:FIXP-almostnear}, (2) from \Cref{lem:efg nfpe circuit,lem:efg polymatrix lcp,lemma:NPP}, (3) from  \Cref{lem:efg nfpe circuit}, and the framework of \citet{Anonymous26:Complexity}, and finally (4) from (2), (3), and \citet{Fearnley23:Complexity}. %
\end{proof}
In the $\PPAD$ case, the algorithm works by finding an $\hat\eps$-sequentially proper equilibrium $\vx^*$, where $\hat\eps$, as determined by \Cref{lemma:NPP}, is only singly exponentially small. But this $\hat\eps$ is not small enough to guarantee that $\vx^*$ is an $\eps$-proper equilibrium even for {\em constant} $\eps$, because \Cref{th:equivalence} would require $\hat\eps$ {\em doubly} exponentially smaller than $\eps$. Indeed, as we discussed in \Cref{sec:hardness}, representing an $\eps$-proper best response even with $\eps=1/2$ constant could require exponentially many bits. Nonetheless, $\vx^*$ can be used to find an {\em exact} sequentially proper equilibrium because of \Cref{lemma:NPP}, and \Cref{th:equivalence} guarantees that exact equilibria coincide, so it does not matter that $\hat\eps$ is relatively ``large'' here. 

\section{Conclusions and future research}

In this paper, we characterized the complexity of computing normal-form proper equilibria in extensive-form games, establishing polynomial-time equivalence to Nash equilibria; this is despite the fact that, as we showed, the classic perturbation scheme of~\citet{Kohlberg86:strategic} is intractable. We find these positive results highly surprising; we originally set out to prove hardness results for this problem. In stark contrast, we showed that a normal-form proper equilibrium in \textit{polytope} games is $\mathsf{NP}$-hard. To our knowledge, this is the first class of games in which the complexity of computing equilibrium refinements does not collapse to that of Nash equilibria. From a broader standpoint, beyond equilibrium refinements, this is the first problem in which equilibrium computation in polytope games is strictly harder relative to extensive-form games.

There are many interesting avenues for future research. While we established \NP-hardness for computing normal-form proper equilibria in polytope games, a precise characterization of its complexity remains open. It would also be interesting to extend our membership results beyond extensive-form games: what properties are necessary and sufficient in polytope games to enable efficient computation of perturbed proper best responses? Finally, \citet{Milgrom21:Extended} recently introduced a refinement of the proper equilibrium which they refer to as \emph{extended proper equilibrium}. This coincides with proper equilibria in two-player games, but beyond that the inclusion is strict. Is the problem of computing extended proper equilibria polynomial-time equivalent to proper equilibria?

\section*{Acknowledgements}

We are indebted to Ratip Emin Berker and Emanuel Tewolde for numerous insightful discussions throughout this project. K.F. thanks Vince Conitzer, whose course on ``Foundations of Cooperative AI'' inspired the initial ideas that led to this work. T.S. is supported by the Vannevar Bush Faculty Fellowship ONR N00014-23-1-2876, National Science Foundation grants RI-2312342 and RI-1901403, ARO award W911NF2210266, and NIH award A240108S001.
    
\bibliographystyle{plainnat}
\bibliography{refs}

\clearpage

\appendix

\section{Quasi-perfect versus normal-form proper equilibria}
\label{sec:qpes}

Here, we elaborate on the difference between QPEs and normal-form proper equilibria. As stated above, QPEs are a superset of normal-form proper equilibria. To see that the inclusion can be strict, it suffices to consider normal-form games, where QPEs are equivalent to normal-form perfect equilibria.

For completeness, we describe an extensive-form game, due to~\citet{Miltersen08:Fast}, in which there is a non-sensible quasi-perfect equilibrium, whereas the sensible equilibrium is the unique normal-form proper equilibrium of the game. The game, referred to as ``matching pennis on Christmas day'' by~\citet{Miltersen08:Fast}, is given in~\Cref{fig:QPEs}. It is a variation of the usual matching pennies games. Both Player 1 and Player 2 select either \textsf{Heads} or \textsf{Tails}. Player 2 wins \$1 if it correctly guesses the action of Player 1. It is a zero-sum game, so Player 2 has diametrically opposing utilities. On top of that, Player 1 has the option of offering a gift of \$1 dollar to Player 2 (\textsf{Gift} in~\Cref{fig:QPEs}). Modulo symmetries, there are $4$ possible outcomes: 
\begin{itemize}
    \item Player 2 guesses correctly and Player 1 offers the gift---Player 2 receives \$2;
    \item Player 2 guesses correctly and Player 1 does not offer the gift---Player 2 receives \$1;
    \item Player 2 guesses incorrectly and Player 1 offers the gift---Player 2 receives \$1; and
    \item Player 2 guesses incorrectly and Player 1 does not offer the gift---Player 2 receives \$0.
\end{itemize}

The only sensible equilibrium in this game is for Player 1 to refrain from offering the gift and each player to play \textsf{Heads} or \textsf{Tails} with equal probability. This is indeed the unique normal-form proper equilibrium of the game~\citep{Miltersen08:Fast}. On the other hand, the following pair of strategies can be shown to be a quasi-perfect equilibrium:
\begin{itemize}
    \item Player 1 plays \textsf{Heads} with probability $\frac{1}{2}$ and \textsf{Tails} with probability $\frac{1}{2}$, and always plays \textsf{NoGift}.
    \item Player 2 plays \textsf{Heads} with probability $\frac{1}{2}$ and \textsf{Tails} with probability $\frac{1}{2}$ at the information set upon observing \textsf{NoGift}, and plays \textsf{Heads} with probability 1 at the other information set.
\end{itemize}

Normal-form proper equilibria refine quasi-perfect equilibria~\citep{vanDamme84:relation}, so from a worst-case perspective, one should always strive for computing a normal-form proper equilibrium. In light of our results, computing a normal-form proper equilibrium is polynomial-time equivalent to computing a quasi-perfect equilibrium.

\begin{figure}[!ht]
    \centering
    \scalebox{0.8}{\begin{tikzpicture}[
    scale=1.3,
    level distance=2cm,
    level 1/.style={sibling distance=6cm},
    level 2/.style={sibling distance=3cm},
    level 3/.style={sibling distance=1.5cm}
]

\node[circle,draw,line width=0.8pt,minimum size=7mm,fill=blue!12] (root) {$1$}
    child {
        node[circle,draw,line width=0.8pt,minimum size=7mm,fill=blue!12] (L1) {$1$}
        child {
            node[circle,draw,line width=0.8pt,minimum size=7mm,fill=red!12] (LL) {$2$}
            child {
                node[rectangle,draw,line width=0.6pt,minimum width=9mm,minimum height=5mm,font=\scriptsize] {$2,-2$}
                edge from parent node[left,draw=none,pos=0.3,font=\scriptsize] {\textsf{Heads}}
            }
            child {
                node[rectangle,draw,line width=0.6pt,minimum width=9mm,minimum height=5mm,font=\scriptsize] {$1,-1$}
                edge from parent node[right,draw=none,pos=0.3,font=\scriptsize] {\textsf{Tails}}
            }
            edge from parent node[left,draw=none,pos=0.3,font=\scriptsize] {\textsf{Gift}}
        }
        child {
            node[circle,draw,line width=0.8pt,minimum size=7mm,fill=red!12] (LR) {$2$}
            child {
                node[rectangle,draw,line width=0.6pt,minimum width=9mm,minimum height=5mm,font=\scriptsize] {$1,-1$}
                edge from parent node[left,draw=none,pos=0.3,font=\scriptsize] {\textsf{Heads}}
            }
            child {
                node[rectangle,draw,line width=0.6pt,minimum width=9mm,minimum height=5mm,font=\scriptsize] {$0,0$}
                edge from parent node[right,draw=none,pos=0.3,font=\scriptsize] {\textsf{Tails}}
            }
            edge from parent node[right,draw=none,pos=0.3,font=\scriptsize] {\textsf{NoGift}}
        }
        edge from parent node[left,draw=none,pos=0.25,font=\scriptsize, yshift=5pt] {\textsf{Heads}}
    }
    child {
        node[circle,draw,line width=0.8pt,minimum size=7mm,fill=blue!12] (R1) {$1$}
        child {
            node[circle,draw,line width=0.8pt,minimum size=7mm,fill=red!12] (RL) {$2$}
            child {
                node[rectangle,draw,line width=0.6pt,minimum width=9mm,minimum height=5mm,font=\scriptsize] {$1,-1$}
                edge from parent node[left,draw=none,pos=0.3,font=\scriptsize] {\textsf{Heads}}
            }
            child {
                node[rectangle,draw,line width=0.6pt,minimum width=9mm,minimum height=5mm,font=\scriptsize] {$2,-2$}
                edge from parent node[right,draw=none,pos=0.3,font=\scriptsize] {\textsf{Tails}}
            }
            edge from parent node[left,draw=none,pos=0.3,font=\scriptsize] {\textsf{Gift}}
        }
        child {
            node[circle,draw,line width=0.8pt,minimum size=7mm,fill=red!12] (RR) {$2$}
            child {
                node[rectangle,draw,line width=0.6pt,minimum width=9mm,minimum height=5mm,font=\scriptsize] {$0,0$}
                edge from parent node[left,draw=none,pos=0.3,font=\scriptsize] {\textsf{Heads}}
            }
            child {
                node[rectangle,draw,line width=0.6pt,minimum width=9mm,minimum height=5mm,font=\scriptsize] {$1,-1$}
                edge from parent node[right,draw=none,pos=0.3,font=\scriptsize] {\textsf{Tails}}
            }
            edge from parent node[right,draw=none,pos=0.3,font=\scriptsize] {\textsf{NoGift}}
        }
        edge from parent node[right,draw=none,pos=0.25,yshift=5pt, font=\scriptsize] {\textsf{Tails}}
    };

\draw[line width=1.0pt,dashed,red!60] ($(LL)+(0,0.15)$) to[bend left=15] ($(RL)+(0,0.15)$);
\draw[line width=1.0pt,dashed,red!60] ($(LR)+(0,0.15)$) to[bend left=15] ($(RR)+(0,0.15)$);

\end{tikzpicture}}
    \caption{The ``matching pennies on Christmas day'' game of~\citet{Miltersen08:Fast}.}
    \label{fig:QPEs}
\end{figure}
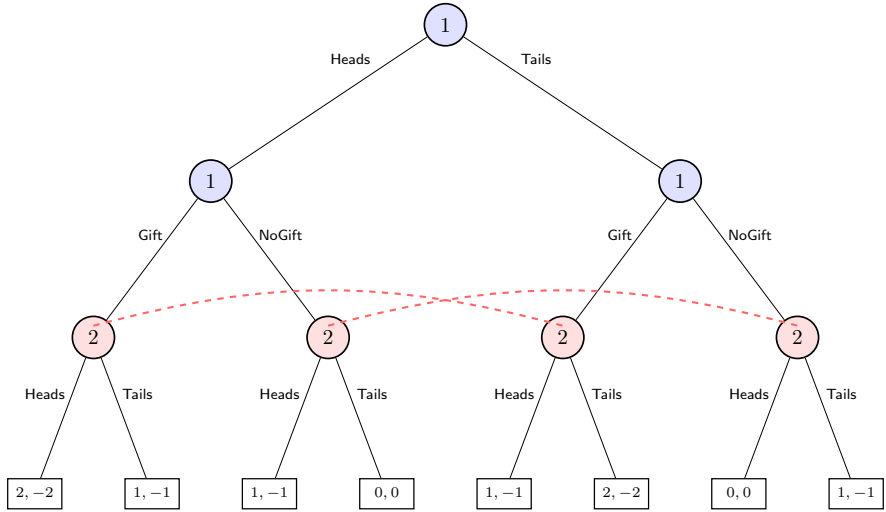
\begin{center}

\end{center}

\section{Relations between different equilibrium refinements}
\label{sec:diagram}

For completeness, we provide a basic diagram (\Cref{fig:diagram}) illustrating the relations between some basic equilibrium refinements in extensive-form games. All notions appearing in~\Cref{fig:diagram} are guaranteed to exist; we do not cover equilibrium concepts that may or may not exist, such as \emph{stable equilibria}~\citep{VanDamme91:Stability}.

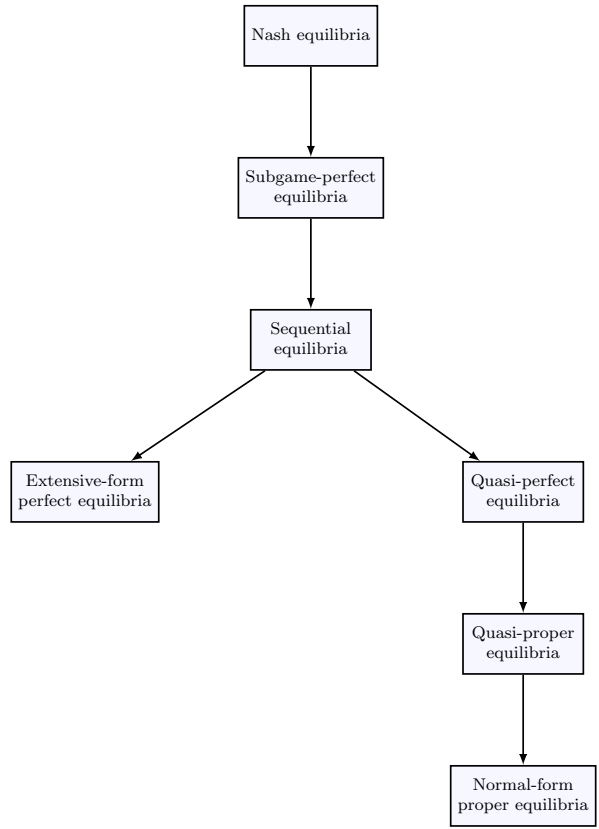
\begin{figure}
    \centering
    \scalebox{0.8}{\begin{tikzpicture}[
    node distance=1.5cm and 1.5cm,
    concept/.style={rectangle,draw,thick,minimum width=2cm,minimum height=1cm,align=center,fill=blue!3,font=\footnotesize},
    arrow/.style={-{latex},thick}
]

\node[concept] (nash) {Nash equilibria};
\node[concept,below=of nash] (subgame) {Subgame-perfect\\equilibria};
\node[concept,below=of subgame] (sequential) {Sequential\\equilibria};
\node[concept,below left=of sequential] (extensive) {Extensive-form\\perfect equilibria};
\node[concept,below right=of sequential] (quasi) {Quasi-perfect\\equilibria};
\node[concept,below=of quasi] (quasiproper) {Quasi-proper\\equilibria};
\node[concept,below=of quasiproper] (proper) {Normal-form\\proper equilibria};

\draw[arrow] (nash) -- (subgame);
\draw[arrow] (subgame) -- (sequential);
\draw[arrow] (sequential) -- (extensive);
\draw[arrow] (sequential) -- (quasi);
\draw[arrow] (quasi) -- (quasiproper);
\draw[arrow] (quasiproper) -- (proper);

\end{tikzpicture}}
    \caption{Relation between different equilibrium refinements in extensive-form games. An arrow $\textsf{A} \to \textsf{B}$ between two equilibrium concepts \textsf{A} and \textsf{B} means that $\textsf{B}$ refines \textsf{A}; that is, every element in \textsf{B} is also in \textsf{A}.}
    \label{fig:diagram}
\end{figure}

\end{document}